\documentclass[runningheads]{llncs}
\usepackage[T1]{fontenc}
\usepackage{graphicx}
\usepackage[utf8]{inputenc}
\usepackage{amsmath}
\usepackage{amssymb}
\usepackage{amsfonts}
\usepackage{amstext}
\usepackage{apptools}
\AtAppendix{\counterwithin{lemma}{section}}
\AtAppendix{\counterwithin{theorem}{section}}
\usepackage{romannum}
\usepackage{tabularx}
\usepackage{xcolor}
\usepackage{hyperref}

\begin{document}
\title{Data Curation from Privacy-Aware Agents\thanks{The work by Roy Shahmoon and Moshe Tennenholtz was supported by funding from
the European Research Council (ERC) under the European Union’s
Horizon 2020 research and innovation programme (grant agreement
740435).}}
%
%
\author{Roy Shahmoon \and
Rann Smorodinsky \and
Moshe Tennenholtz}
\authorrunning{R. Shahmoon et al.}
%
\institute{Technion - Israel Institute of Technology, Haifa
3200003, Israel
\email{shroy@campus.technion.ac.il\\
\{rann,moshet\}@technion.ac.il}}

%
\maketitle              
\begin{abstract}
A data curator would like to collect data from privacy-aware agents. The collected data will be used for the benefit of all agents. Can the curator incentivize the agents to share their data truthfully? Can he guarantee that truthful sharing will be the unique equilibrium? Can he provide some stability guarantees on such equilibrium? We study necessary and sufficient conditions for these questions to be answered positively and complement these results with corresponding data collection protocols for the curator.
Our results account for a broad interpretation of the notion of privacy awareness. 

\end{abstract}
%
%
%





    

\section{Introduction}

Consider a society of agents who would like to compute some aggregate statistic of their online activity. For example, their average daily consumption of online games, the median number of messages they send, or their daily spend on Amazon. To do so they share information with some central authority. The way agents share information is by downloading an add-on to their browser which collects and reports the data. They can choose to download an add-on which accurately reports the data (hereinafter, the `accurate' add-on), an add-on which completely mis-represents the data, or anything in between.

The agents have privacy concerns and so may not be willing to download the accurate add-on. A data curator (also referred to as the `center'), who possesses some data as well, collects the agents' data and uses it to make some computations which she will then share privately back with the agents. We ask whether the curator can propose a scheme for sharing her computations which will induce the agents to download the accurate add-on. Technically, we are interested in schemes where there is a {\bf unique} equilibrium in which agents download the accurate add-on.


More abstractly, we study a setting where $N$ agents and a data curator plan ahead of time how to share information in order to compute a function of the information.
If the agents have no privacy concerns then the curator can simply ask the agents for their information, compute the function and distribute the outcome back. We are interested in the setting where agents actually have privacy concerns. We capture the agents concerns by a privacy cost function. This cost is a function of how accurate is the information they reveal. In our simple setting, the private information of the agents is binary and privacy cost is modelled as an arbitrary increasing function of the probability they will disclose the actual value of this private bit.%
\footnote{In fact, for our results we do not require monotonicity but rather that privacy cost is highest when the bit is fully revealed and lowest when no information is disclosed.}
 
The game between the agents and the curator takes place ahead of time, before any of the information becomes known to the agents. Each agent submits an information disclosure protocol which, in our binary setting, is parameterized by two numbers - the probabilities of disclosing the actual bit for each of the two possible outcomes.
 
The curator also submits her protocol, which discloses which information she will share back with each of the agents depending on actual vector of protocols (the vector of pairs of probabilities). The timeline of the game is as follows. The curator proposes a protocol for sharing information and then the N agents submit their own protocols simultaneously. Only then the private information is revealed to the agents and from then on information passes back and forth between the agents and the curator according to the protocols that they have all designed in advance. Note that agents play  no active role in the last stage as they cannot affect the protocols during run-time.

Our objective is to identify necessary and sufficient conditions on the parameters of the problem that will allow the curator to propose a protocol which, in turn, will `guarantee' that the agents, in-spite of having privacy concerns, will reveal their information accurately (in the jargon of our first paragraph - will download the accurate add-on) and consequently the function will be computed accurately by all. In other words, where complete information revelation is a unique equilibrium.%
\footnote{Requiring a unique equilibrium makes our problem technically more challenging and conceptually more compelling than merely requiring existence of a truthful equilibrium, as we demonstrate in the paper.}

The curator's protocol induces a game among the agents.
We will deem the protocol as successful if there exists an equilibrium where agents reveal their information accurately and the correct value of the function becomes known to all. Our goal is to identify conditions for the existence of a successful protocol with a unique equilibrium.

Indeed we identify such conditions. First, we define the notion of a (privacy-) fanatic agent as an agent whose gain from knowing the value of the function is overshadowed by the dis-utility from disclosing her private information (Definition \ref{Fanatic Definition}). We show that a necessary and sufficient condition for the existence of a protocol with a truthful equilibrium (not necessarily unique) is that no agent is fanatic (Corollary \ref{Corollary 4}).

Second, to guarantee the uniqueness of the aforementioned equilibrium we argue that the curator must herself have private information which is part of the input and, furthermore, this information must not be negligible. We refer to curators whose information is negligible as `helpless' (Definition \ref{Helpless Definition}) and prove  that a necessary and sufficient condition for the existence of a successful protocol which equilibrium is unique, is for the agents not to be fanatic and for the curator not to be helpless (Corollary \ref{Corollary 4}, Lemma \ref{Lemma 2}).  

In addition,  we go ahead and construct such protocol, which we refer to as the {\it competitive protocol}. The {\bf unique equilibrium} of this protocol is one where agents forego their privacy concerns and information is passed on accurately. Our protocol is optimal among all successful protocols in the following sense. Assume that for some parameters of the problem there exists a protocol possessing an equilibrium that supports the accurate computation of the function. Then, indeed, on these parameters our protocol has this exact property (Theorem \ref{Optimality Theorem}).

This protocol has a disturbing caveat. Its single equilibrium is not in dominant strategies. In fact, off equilibrium it may be the case that a cooperative agent is guaranteed nothing. We then go on and introduce an alternative optimal protocol with a unique truthful equilibrium, which provides some natural guarantees to cooperative agents, independently of the behavior of others.

 \section*{Related Work}
 
Our work employs strategies that are conditional commitments. In this regard our work is related to work on multi-sender single-receiver conditional commitments. The setting of conditional ex-ante commitments which was initially introduced for a single sender \cite{kamenica_gentzkow_2009} has been extended to the context of multiple senders \cite{GKmultisender}.  In particular, it is shown that a greater competition tends to increase the amount of information revealed. Our main protocol also  exploits competition among data owners to facilitate data contribution.

    Among work in (mediated) multi-party computation the work most related to ours is by  Smorodinsky and Tennenholtz \cite{smorodinsky2006overcoming}, later extended by Aricha and Smorodinsky \cite{aricha2013information}. In this line of research the authors address the question of multi-party computation where agents would like to free-ride on others' contributions. 
    An important aspect of their article is that they address the agents sequentially (and not simultaneously as in our work). Another difference is that while in our model the agents wouldn't like to share their values due to privacy concerns, the agents in their model incur a cost for retrieving their values. Thus, in their model the agents' cost is discrete (it is either zero or the cost of retrieval) but in ours it is continuous (by sharing partial data). Similar to our work, they aim to find a mechanism which overcomes the free-riding phenomenon and convinces the agents to send true values for correct calculation of the social function. 
    
    Another genre of mediated multi-party computation focuses on  getting exclusive knowledge of the computation result rather than considering privacy considerations. The first work on that topic, referred to as non-cooperative computation (NCC) \cite{shoham2005non}, considers a setting where each agent wishes to compute some social function by declaring a value to a center, who reports back the results to the agents after applying the function. In the extensions of the basic model the authors also  give a greater flexibility to the center. The results of the paper show conditions which classify the social functions as deterministic NCC, probabilistic NCC and subsidized NCC in both correlated and independent values settings. In addition, Ashlagi, Klinger and Tennenholtz \cite{ashlagi2007k} extend the study of NCC functions to the context of group deviations. 
    Notice that in our work, the incentives of the agents are correctness and privacy, in contrast to the NCC setting which focuses on correctness and exclusivity (the desire that other agents do not compute the function correctly) as main incentives.     However,  McGrew, Porter and Shoham \cite{mcgrew2003towards} extended the framework of \cite{shoham2005non}, focusing on cryptographic situations. The utility of the agents in their paper is lexicographic and relies on four components: correctness, privacy (similarly to our model), exclusivity and Voyeurism (the wish to know other agents’ private inputs). They show properties of the social functions (such as reversibility or domination) that affect the existence of an incentive-compatible mechanism in each ordering of the utility. The utility function we consider is not lexicographic. In fact, their work does not consider different privacy levels, and refers to privacy only through whether agent $j$ knows agent $i$'s value with certainty. There is also no importance in their work for hiding value from the center.    
    Besides the above central differences, in our model the center possesses her own private value (in contrast to theirs) and we also allow the center much more flexibility in the protocol (not necessarily returning the result of the social function's calculation).     
    
    Much work has been devoted to variants of the  non-mediated multi-party computation problem. For example, Halpern and Teague \cite{halpern2004rational} approach the problems of secret sharing and multiparty computation. In their paper, the agents would like to know the social function's value but also prefer exclusivity. They focus on finding conditions on the running time of the multiparty computation. Specifically, they show that there is no mechanism with a commonly-known upper bound on running time for the multiparty computation of any non-constant function. In contrast to our work, the mediator is replaced by cryptographic techniques and there is no explicit interest in privacy. The emphasis is mainly on removing the reliable party (as developed in \cite{AbrahamDGH06,AbrahamDGH19}).
    
    Another setting with some similarity to ours is by Ronen and Wahrmann \cite{ronen2005prediction} who introduced the concept of prediction games. In prediction games there are agents with private information and a center whose goal is to compute the value of a social function ahead of its realization (and hence the name {\bf prediction} games). Whereas, in their model the center can condition his payoffs to the agents on the realization of the social function this is not the case here, where the center can only use the private information of all others (including herself) to incentivize each agent. Two additional differences are as follows:
    First, the agents' costs are discrete, not continuous (the agents incur a cost by accessing the information, without privacy concerns). Second, the agents are approached serially (not simultaneously).

    
    Similarly, models of information elicitation, and in particular ones that consider the decisions to be made following the aggregation of information \cite{ChenK11} are relevant to our work. These however focus on the use of monetary transfers 
    for eliciting information about uncertain events, rather than on non-monetary mechanisms for function computation based on distributed inputs and conditional commitments, nor they refer to privacy.   
    
    The papers \cite{ghosh2015selling,cummings2015accuracy,chen2018optimal} address a remotely related domain of data markets for privacy, where a data analyst wishes to pay as little as possible for data providers in order to obtain an accurate estimation of some statistic.
    
    Our setting relates to mechanism design without money, where the center may use her private information in order to incentivize the agents to reveal their secrets \cite{procaccia2013approximate}.
    
    Notice that our notion of privacy differs from the one used in the context of differential privacy \cite{DworkMNS06,DworkRoth}. The rich literature on differential privacy is concerned with protocols where agents' private secrets should not be revealed to other agents. However, the privacy concerns that agents have in our model is vis-a-vis the information revealed to the center.
    
    Another related setting is (silo-)federated learning \cite{kairouz2021advances}. In that setting, a small number of firms share data to obtain some joint computation. When these firms have privacy costs then we get to a setup that resembles the one we study.


    
    
    
    

\section{The Model}

    In our multi party computation setting we assume a society of $N$ agents and a center (denoted as agent 0), each has a private information (taking the form of a binary bit). By revealing these secrets the agents derive dis-utility (due to privacy concerns).
    The center has no privacy concerns and she is the only agent who can communicate with the others and facilitate the desired computation of a binary social function $g : \{0,1\}^{N+1} \to \{0,1\}$.
    The flow of the game is as follows:
    
    \begin{enumerate} 
    \item First, the center declares a communication protocol between herself and a set of agents. The protocol is a function that receives the agents' strategies (conditional commitments) and decides how to operate accordingly.
    \item
    Then, each agent commits to a strategy which takes the form of a distribution over a set of messages as a function of her private bit, and sends it to the center. Commitments are made  simultaneously.      Notice that the center is obligated to the protocols and the agents are obligated to their strategies.
    \item
    Afterwards, agents receive their bits simultaneously, random messages are drawn according to the aforementioned strategies and then communicated to the center.
    \item
    Finally, the center's responses are generated according to the vector of agents' inputs and each agent receives her own response. Based on this response each agent guesses the true value of the social function and derives a (non-negative) profit if and only if the guess is correct.
    \end{enumerate}
    We assume that the agents are rational individuals trying to maximize their own expected utility. The center's goal is to uncover the true value of the social function.
    We now introduce some formal notations.

    Let $\Omega$ be the sample space of the game where a typical instance is $b \times m \times f$ such that: 
    \begin{itemize}
        \item $b=(b_0,\dots,b_N)\in\{0,1\}^{N+1}$ is a vector containing the bits of the center and the agents.
        \item $m=(m_0,\dots,m_N) \in \{0,1\}^{N+1}$ is a vector of the agents' messages to the center ($m_0=b_0$ is the center's message to himself).
        \item $f=(f_1,\dots,f_N)$ is a vector of the center's replies to each agent, such that $\forall i\in \{1,\dots,N\}: f_i\in \{0,1\}^{N+1}$.
    \end{itemize}
    
    Note that each message from an agent (including the center) to the center is a bit, but each reply from the center to an agent (excluding the center) is a vector containing $N+1$ elements. These elements should be interpreted as information on the $N+1$ original bits of all the agents. 
    
    We assume that the bits are $i.i.d$ and distributed uniformly over $\{0,1\}$.\footnote{The
    model can be generalized to any distribution over $\{0,1\}$ while maintaining the same results. The main change would be in the set of zero-information strategies (which is defined later in this page).}
    
    For each agent $i\in\{1,\dots,N\}$ we define her  strategy space (possible conditional contracts) as follows:
    \[S_i=\{(p_i,q_i)|\;p_i,q_i\in[0,1]\}\]
    where 
    \[p_i=p(m_i=0|b_i=0), q_i=p(m_i=0|b_i=1).\]
    
    A \emph{strategy profile} $\vec{s}=(s_1,\dots,s_N)$ is a set of the agents' strategies, such that $s_1 \in S_1,\dots,s_N \in S_N$.\\
    The strategy $s_i=(1,0)$ fully reveals the agent's secret and hence is called \emph{truthful} and denoted by $s^t$.
    Any strategy $s_i = (p_i,q_i)$ such that $p_i = q_i$ reveals nothing about the agent's secret and hence is called a 
    \emph{zero-information} strategy. We denote by $S^r = \{(c,c)\;|\; c\in [0,1]\}$ the set of zero-information strategies.
    
    For each agent $i \in \{1,\dots,N\}$ and for each agent $j \in \{0,\dots,N\}$, let $F_{i,j}$ be the center's protocol for agent $i$ about agent $j$.
    \[F_{i,j}:S_1\times S_2\times \dots \times S_N \times m_j \to \Delta(\{0,1\}).\]
    
    Let $F_i=(F_{i,0},\dots,F_{i,N})$ be the center's protocol for agent $i$ such that
    \[F_i:S_1\times S_2\times \dots \times S_N \times m_0 \times \dots \times m_N \to \Delta(\{0,1\})^{N+1}.\]

    
    The protocol $F_i$ for agent $i$ determines the probability distribution over the messages to $i$ as a function of all agents' conditional commitments. 
    Notice that each message from the center to an agent is a tuple of $N+1$ bits (e.g. standing as place holders that can serve to provide information about all participant bits).
    \footnote{We can reduce the communication complexity of the protocol by having the center return two bits (instead of $N+1$) to each agent. To do so the center simulates agents' computations for the estimate of the social function in each of the agent's possible bits. Nevertheless, we maintain the current presentation as it is more transparent.}
    
    
    Let $F=(F_1,\dots,F_N)$.
    Notice that the prior distribution over the bits together with the protocol $F$ and the strategy profile $\vec{s}$  induce a probability distribution over the sample space $\Omega$. Knowing the value of her own bit $b_i$, her message $m_i$, and the center's reply $f_i$, the agent computes the conditional probability $p(g(b)=1|b_i,m_i,f_i)$ and uses it to submit her guess, $a_i$, over the value of $g$ as follows:
     \[a_i(b_i,m_i,f_i)=
        \begin{cases}
          1 \hspace{2.6cm} \text{ if } p(g=1|b_i,m_i,f_i) > p(g=0|b_i,m_i,f_i)\\
          0 \hspace{2.6cm} \text{ if } p(g=1|b_i,m_i,f_i) < p(g=0|b_i,m_i,f_i)\\
          random\{0,1\} \hspace{0.7cm} \text{ if } p(g=1|b_i,m_i,f_i) = p(g=0|b_i,m_i,f_i)
        \end{cases}
    \]
    
    
    Let as denote the profit of agent $i$:
    \[v_i(b,m_i,f_i)= 
        \begin{cases}
          c_i \hspace{0.5cm} \text{ if } a_i(b_i,m_i,f_i) = g(b)\\
          0 \hspace{0.6cm} \text{ if } a_i(b_i,m_i,f_i)\neq g(b)
        \end{cases}
    \]
    where $c_i>0$ is an arbitrary constant.
    
    For a given protocol $F$, the vector of strategies $\vec{s}=(s_1,\dots,s_N)$ induces the expectation of $v_i$, denoted as $V_i(\vec{s};F)$.
    Note that $V_i(\vec{s};F)$ doesn't depend on $F$ but only on $F_i$ and from now on we will write $V_i(\vec{s};F_i)$ rather than $V_i(\vec{s};F)$. 
    
    We shall refer to the function
    \[V_i(S_i,S_{-i};F_i):\;S_i \times S_{-i} \times F_i \to \mathbb{R_+}\]
    as the \emph{profit function} of agent $i$ in the protocol $F_i$.

    
    Let $price_i$ be the privacy cost function of agent $i$.
    \[price_i: S_i \to \mathbb{R_+}\]
    We assume that an agent incurs a zero privacy cost if and only if she plays a zero-information strategy. Furthermore, truthfulness incurs the highest price.\footnote{It seems natural to further assume that $price_i$ is an increasing function of $p_i$ and a decreasing function of $q_i$ but it is not necessary of our results.}
    Formally:
    \[\forall i\in \{1,\dots,N\}, \forall s^r\in S^r, \forall s_i\in S_i-S^r: price_i(s^t) \geq price_i(s_i) > price_i(s^r) = 0.\]
    In addition, we assume that $price_i$ is continuous.
    
    
    We define the \emph{utility function} of an agent $i$ in the protocol $F_i$ as:
    \[U_i(S_i,S_{-i};F_i):\;S_i \times S_{-i} \times F_i \to u_i\in\mathbb{R}\]
    \[s.t.\hspace{0.3cm} U_i(S_i,S_{-i};F_i)=V_i(S_i,S_{-i};F_i)-price_i(S_i).\]

    

     Let $F=(F_1,\dots,F_n)$ be a protocol. Recall that the agents choose a strategy before they receive their private information. We say that a strategy profile $\vec{s}=(s_1,\dots,s_N)$ is an equilibrium in $F$ if \[\forall i\in \{1,\dots,N\}, \forall s'_i\in S_i: \; U_i(s_i,s_{-i};F_i) \geq U_i(s'_i,s_{-i};F_i)\]

    Let us recall the flow of the game using the new notations:
    \begin{itemize}
        \item
    The center declares a protocol $F=(F_1,\dots,F_N)$.
    \item
    Each agent $i \in \{1,\dots,N\}$ chooses a strategy $s_i\in S_i$ and immediately incurs a cost of $price_i(s_i)$.
    \item
    The private bit $b_i$ of each agent $i$ is realized and the message $m_i=0$ is sent to the center with the probability $p_i$ whenever $b_i=0$ and with the probability $q_i$ whenever $b_i=1$. Otherwise, the message $m_i=1$ is sent with the complementary probabilities.
    \item
    Each agent $i$ receives a reply $f_i \in \{0,1\}^{N+1}$ according to the distribution $F_i(\vec{s},m) = (F_{i,0}(\vec{s},m_0),$ $\dots,F_{i,N}(\vec{s},m_N))$ where $\vec{s}=(s_1,\dots,s_N)$.
    \item
    Finally, each agent $i$ leverages her information $b_i,m_i,f_i,\vec{s}$ by applying Bayes rule to compute her best guess of the value of $g(b)$. If she guesses correctly she incurs a profit, otherwise she doesn't. Her overall expected utility is $U_i(\vec{s}=(s_1,\dots,s_N);F_i) = V_i(\vec{s}=(s_1,\dots,s_N);F_i) - price_i(s_i)$.
    
    \end{itemize}
    
    \subsection{A Preliminary Result}
    
    
    
    
    After receiving the center's reply, each agent has to decide whether the true value of $g$ is 0 or 1 based on the information she has ($b_i,m_i,f_i$). Therefore, the expected profit will be the maximum between the alternatives from the agent's point of view at the time of the decision.  
    
    Now, we prove that the expected profit of an agent is the following intuitive expression that will serve us later in the proof of Lemma \ref{Comparison Lemma}:
    \begin{lemma} \label{Profit Lemma}
        \[V_i(\vec{s};F_i) = c_i\cdot \underset{b_{-i},m_{-i}}{E}[\underset{b_i,f_i,m_i}{E}[max\{p(g(b)=0|f_i,b_i,m_i),p(g(b)=1|f_i,b_i,m_i)\}]].\]
    \end{lemma}
    
    The proof is quite technical and therefore appears in Appendix A.

    \section{Informational Profitability}
    
    In this section we address the concept of "more / better" information and show its relation to an agent's profit.
    Our goal is to make a connection between the agent's knowledge and her profit, and particularly to compare different replies given by the center.
    We prove here that there are cases in which an agent prefers in expectation one response over another. Understanding the influence of the center's replies on the agents' profits is a useful technical tool for the results presented in the following section.  
    
    Let $i\in \{1,\dots,N\}$ be an agent, let $F_i^1,F_i^2$ be two different protocols for agent $i$, let $\vec{s_1}=(s_1^1,\dots,s_N^1),\; \vec{s_2}=(s_1^2,\dots,s_N^2)$ be strategy profiles and let $m^1=(m_0^1,\dots,m_N^1), m^2=(m_0^2,\dots,m_N^2)$ be the agents' messages to the center according to $\vec{s_1},\vec{s_2}$ respectively (note that $m_0^1=m_0^2=b_0)$.
    Let $f_i^1, f_i^2$ be the center's replies drawn from the distributions  $F_i^1(\vec{s_1},m^1), F_i^2(\vec{s_2},m^2)$ respectively.
    In addition, assume that $y^1=(y_0^1,y_1^1,\dots,y_N^1)$ represents in each coordinate $k$ the probability that the $k^{th}$ agent has the value $1$ according to $i$'s inference from $f_i^1,m_i^1,b_i$. Formally:
    \[\forall k \in \{0,\dots,N\}: y^1_k = p(b_k = 1| f_i^1,m_i^1,b_i) = p(b_k = 1| f_{i,k}^1,m_i^1,b_i).\]
    The last equality is derived from the fact that $b_k$ and $f_{i,-k}^1$ are independent (recall that according to the definition of a protocol, $f_i^1=(f_{i,0}^1,\dots,f_{i,N}^1)$, such that each component in $f_i^1$ gets as input a message only from the corresponding agent).
    Similarly, assume that $y^2=(y_0^2,y_1^2,\dots,y_N^2)$ represents in each coordinate $k$ the probability that the $k^{th}$ agent has the value $1$ according to $i$'s inference from $f_i^2,m_i^2,b_i$. Formally:
    \[\forall k \in \{0,\dots,N\}: y^2_k = p(b_k = 1| f_i^2,m_i^2,b_i) = p(b_k = 1| f_{i,k}^2,m_i^2,b_i).\]
    Let $Y^1=(\underset{f_i^1,m_i^1,b_i}{E}[y^1_0],\dots,\underset{f_i^1,m_i^1,b_i}{E}[y^1_N]),\; Y^2=(\underset{f_i^2,m_i^2,b_i}{E}[y^2_0],\dots,\underset{f_i^2,m_i^2,b_i}{E}[y^2_N])$.
    
    \begin{lemma}[Comparison Lemma] \label{Comparison Lemma}
       Let $i,j\in N$, where $i \neq j$.
        Assume that $Y^1_{-j} = Y^2_{-j}$ and $|Y_j^1-0.5| \geq |Y_j^2-0.5|$.
        Then,
        $V_i(\vec{s_1};F_i^1) \geq V_i(\vec{s_2};F_i^2)$.
    \end{lemma}

    \paragraph*{\underline{Intuition}:}\mbox{}\\
    An agent may have two different responses from the center $f_i^1,\:f_i^2$.
    She would like to know whether $f_i^1$ is better or worse than $f_i^2$ in terms of profit.
    Each response reveals some information about the agents' bits ($y^1$ or $y^2$) and therefore about the true value of $g$.
    Our observation relies on the distances of $y_j^1,\:y_j^2$ from $0.5$. Suppose that $y_j^1$'s distance from $0.5$ is greater than $y_j^2$'s distance from $0.5$, then $y_j^1$ reveals "more" information on $b_j$ and therefore on $g$ as well. As an intuitive example, let $y_j^1=0$ and let $y_j^2=0.5$. Then, according to $y_j^1$ we infer that $b_j=0$. On the other hand, according to $y_j^2$ we infer that $p(b_j=0|y_j^2)=p(b_j=1|y_j^2)=0.5$.
    Hence, $y_j^1$ is more informative than $y_j^2$.
    Therefore, one may expect that knowing $y_j^1$ over $y_j^2$ will also lead to a better guess on the social function's value. 
    But even though the intuition is quite clear, the claim is not immediate and true only in expectation. The proof of the Comparison Lemma appears in Appendix B. 
    
    We now approach the general case,  by applying the previous lemma several times until we get the desired result. 

    \begin{theorem}[Comparison Theorem] \label{Comparison Theorem}
    Assume that $\forall k\in \{0,\dots,N\}: |Y_k^1-0.5| \geq |Y_k^2-0.5|$. Then,
    $V_i(\vec{s_1};F_i^1) \geq V_i(\vec{s_2};F_i^2)$.
    \end{theorem}

    The proof of the Comparison Theorem appears in Appendix B as well.

    


    \section{The Competitive Protocol}

   This is the main section of the paper where we present the existence of a protocol that truthfully implements in unique equilibrium the computation of social functions, as well as prove its optimality.    
    We show that under some minor natural assumptions we achieve a protocol, titled as the {\em competitive protocol}, with the following properties:
    \begin{enumerate}
    \item Truthfulness is its only equilibrium.
    \item When the agents play according to the equilibrium, the center knows the true value of $g$.
   \item When the agents play according to the equilibrium, they know the true  value of $g$.
    \item This protocol is optimal - if there exists a protocol with a truthful equilibrium (for some fixed utility and social functions) then it is also a unique equilibrium of the competitive protocol. 
    \end{enumerate}

    We start with a definition. The \emph{relative price} $price_i^r: S_i \to [0,1]$ of agent $i$ is \[price_i^r(s_i) = \frac{price_i(s_i)}{price_i(s^t)}.\]
    Therefore,
    \[\forall i\in \{1,\dots,N\}, \forall s^r\in S^r, \forall s_i\in S_i-S^r: 1 = price_i^r(s^t) \geq price_i^r(s_i) > price_i^r(s^r) = 0.\]
    
    Recall that the function $price_i$ is continuous and therefore the relative price is continuous as well.

    \begin{definition} [The competitive protocol]
    
    If all agents play $s^t$ they receive $m=b$.
    Else, the agent whose relative price is the highest (comparing to the others) receives $m=(b_0,m_1,\dots,m_N)$ while the others receive uniformly distributed random bits.
    If at least two agents incur the same highest relative price then everyone receives random bits.\footnote{The competitive protocol satisfies the properties above in case that $N\geq2$. If $N=1$ then it is easy to verify that a simple protocol that shares the center's bit if and only if the agent is truthful is the only protocol that reveals the true value of the social function for both the center and the agent.}$^,
    $\footnote{A viable criticism against the proposed protocol is that it requires that the center  is familiar with the agents' privacy functions. However, if the center replaces the real price function with the indicator price function - assigning the price zero whenever the agent uses a zero-information strategy and one otherwise - this criticism no longer holds, yet the new protocol will have the exact same equilibria as the original one.}
    \\\\$\forall i\in \{1,\dots,N\}, \forall s_1\in S_1,\dots,s_N\in S_N, \forall m_{-0}\in \{0,1\}^N:$\[ F_i^c(s_1,\dots,s_N,m)=
    \begin{cases}
    m & \text{if } s_1=\dots=s_N=s^t\\
    m & \text{if } price_i^r(s_i) > \underset{j \neq i}{max}\{price_j^r(s_j)\} \\
    \overrightarrow{0.5} & \text{otherwise}
    \end{cases}\]
    Note that the notation of $\overrightarrow{0.5}$ refers to uniformly distributed random bits.
    \end{definition}

    \paragraph*{\underline{Intuition}:}\mbox{}\\
    First, we want to construct a protocol with truthfulness as its only equilibrium. In order to implement it successfully, we make sure that in any other strategy profile there is at least one unsatisfied agent who prefers to deviate. We try to guarantee this property by creating a competition between the agents in every situation except from the truthful equilibrium. The competition is based on the relative prices that the agents incur. An agent who pays the highest relative price gets everything the center may offer while the others get nothing useful.
    For instance, assume a situation with two agents, each wants to get the center's information and pay as little as possible. If agent 1 knows agent's 2 strategy, then she knows agent's 2 relative price so she will pick a strategy with a slightly higher price in order to "win" and get the center's info. Agent 2 has the same considerations, yielding a situation that each prefers to slightly change her strategy trying to overcome the other agent.
    They both may be satisfied when they play truthfully and get all the information.\footnote{This process is similar to the concept of Bertrand competition from the field of economics.}
    
    
    One may wonder why we suggest a complex protocol where there is an intuitive protocol that seems to deliver the same results.
    The initial idea that comes to mind is a protocol that suggests all of the center's information $(b_0,m_{-0})$ for truthful agents and nothing ($\overrightarrow{0.5}$) otherwise. One may believe that the strictness of the protocol is the best way to incentivize the agents to be honest. 
    Unfortunately, this straightforward approach has a major flaw that the competitive protocol circumvents - it allows an additional equilibrium, one where zero-information strategies are played by all agents. Indeed, an agent may not be interested in paying a full privacy cost in exchange for the center's private information, by deviating to truthfulness.
    Furthermore, such zero-information equilibrium may be more stable than the truthful equilibrium in various settings, in the following sense: a deviation of one agent from the truthful equilibrium may lead others deviate as well, possibly leading to the zero-information equilibrium.
    
    In contrast, such a zero-information strategy profile does not form an equilibrium in the competitive protocol: any single agent can profitably deviate by suffering a negligible cost and obtain the center's private information. Indeed, as we will show, under natural conditions truthfulness is a unique equilibrium in our protocol.   

    \subsection{Uniqueness}
    
    We now turn to show that truthfulness is the unique equilibrium of the competitive protocol.  
    Our proof will make use of the following definition of protocol improvement and its property in the Lemma following it, building on the technical comparison Theorem presented in the previous section.

    Suppose that $F_i$ is a protocol and $\vec{s}$ is a strategy profile. Recall that the first coordinate of $F_i(\vec{s},m)$ (i.e. $F_{i,0}(\vec{s},m_0)$) may reveal information only about the center's bit.
    
    Let $F_i$ be a protocol for agent $i$ and let us define the improved protocol, denoted by $F_i^+$, which returns in the first coordinate the center's bit if $s_i$ is not a zero-information strategy.
    Formally: 
    
    \begin{definition} \label{Improvement Definition}
    Let $F_i=(F_{i,0},\dots,F_{i,N})$ be a protocol. Hence,
    \begin{multline*}
        \forall s_1\in S_1,\dots, s_N\in S_N,  \forall m_{-0}\in \{0,1\}^N:\\ F_i^+(s_1,\dots,s_N,m)=
    \begin{cases}
    F_i(s_1,\dots,s_N,m) & \text{if } s_i\in S^r\\
    (m_0,F_{i,-0}(s_1,\dots,s_N,m_{-0})) & otherwise
    \end{cases}
    \end{multline*}
    \end{definition}

    Recall that $S^r$ is the set of zero-information strategies.
    
    For every strategy profile $\vec{s}$, $F_i$'s improvement ($F_i^+$) fully reveals the center's bit $m_0=b_0$ instead of giving $F_{i,0}(\vec{s},m_0)$ (which may contain a partial or full information on $b_0$) when the agent does not play a zero-information strategy.
    
    Now we approach two technical lemmas that their proofs appear in Appendix C. Lemma \ref{Lemma 1} shows that the agents' utility from an improved protocol is at least as high as their utility from the original protocol. Lemma \ref{Better than Nothing} shows that $\overrightarrow{0.5}$ is the worst reply that an agent may get.

    \begin{lemma} \label{Lemma 1}
        For every agent $i$ and her corresponding protocol $F_i$:
        \[\forall s_1\in S_1,\dots, s_N\in S_N: U_i(s_1,\dots,s_N;F_i^+) \geq U_i(s_1,\dots,s_N;F_i).\]
    \end{lemma}
   


    

    \begin{lemma} \label{Better than Nothing}
        Let $F_i^1,F_i^2$ be protocols for agent $i$. Let $\vec{s_1},\vec{s_2}$ and $m^1,m^2$ be the strategy profiles and the agents' messages in those protocols respectively. If $F_i^2(\vec{s_2},m^2)=\overrightarrow{0.5}$ then $V_i(\vec{s_1};F_i^1) \geq V_i(\vec{s_2};F_i^2)$.
    \end{lemma}
    
    Now we approach the assumption that we have mentioned before.

    \begin{definition} \label{Helpless Definition}
    We call the center \emph{helpless} if for every agent $i$ and every protocol $F_i$:
    \[\forall s_i \in S_i-S^r, \forall s_1^r,\dots,s_N^r\in S^r:
    U_i(s_i^r,s_{-i}^r;F_i) \geq U_i(s_i,s_{-i}^r;F_i^+).\]
    \end{definition}

    The center is helpless if she doesn't have the capability to prevent an equilibrium of zero-information strategies.
    In case of a zero-information strategy profile, all the center may offer is her bit. If there isn't any agent in this profile who prefers to deviate from a zero-information strategy in order to receive the center's information, then we call the center helpless.

    \begin{example} \label{Helpless Example}
    Suppose that $N=2$ and let $g(b)=\begin{cases}
    b_0 & \text{if} \hspace{0.2cm} xor(b_1,b_2)=0\\
    \bar{b_0} & \text{if} \hspace{0.2cm} xor(b_1,b_2)=1
    \end{cases}$.\\
    
    The reader may verify that the center is helpless. In a situation of a zero-information strategy profile notice that the center's bit is not useful for the agents without information on the other agent's bit. Therefore, regardless of the center's protocol the agents won't deviate from a zero information strategy in this profile.
    Note that in this example the center's information is essential for the exact computation of $g$ even though she is helpless.
    \end{example}

    \begin{lemma} \label{Lemma 2}
        The center is helpless if and only if all the protocols have the equilibria $(s_1^r,\dots,s_N^r) \; \forall s_1^r,\dots,s_N^r \in S^r$.
    \end{lemma}

    \begin{proof}[Direction 1]
    Assume that the center is helpless, and let $F_i$ be an arbitrary protocol for agent $i$.
    We show that the profile $(s_1^r,\dots,s_N^r)$ is an equilibrium in $F_i$ for all $s_1^r,\dots,s_N^r \in S^r$:
    \[    \forall s_i \in S_i-S^r, \forall s_1^r,\dots,s_N^r\in S^r:\;
    U_i(s_i^r,s_{-i}^r;F_i) \geq
    U_i(s_i,s_{-i}^r;F_i^+) \geq
    U_i(s_i,s_{-i}^r;F_i).\]
    Thus, agent $i$ doesn't want to deviate to any $s_i \in S_i-S^r$.\hfill\(\blacksquare\)
    
    \end{proof}

    \begin{proof}[Direction 2]
    Suppose that all the protocols have the equilibria $(s_1^r,\dots,s_N^r)$ for all $s_1^r,\dots,s_N^r \in S^r$, and assume in contradiction that the center is not helpless. Therefore, there is an agent $i$ and a protocol $F_i$ such that $\exists s_i \in S_i-S^r, \forall s_1^r,\dots,s_N^r\in S^r: U_i(s_i,s_{-i}^r;F_i^+) > U_i(s_i^r,s_{-i}^r;F_i).$
    But $\forall s_1^r,\dots,s_N^r\in S^r: U_i(s_i^r,s_{-i}^r;F_i^+) = U_i(s_i^r,s_{-i}^r;F_i)$
    and therefore,
    $\exists s_i \in S_i-S^r, \forall s_1^r,\dots,s_N^r\in S^r: U_i(s_i,s_{-i}^r;F_i^+) > U_i(s_i^r,s_{-i}^r;F_i^+)$
    meaning that $\forall s_1^r,\dots,s_N^r\in S^r: (s_1^r,\dots,s_N^r)$ is not an equilibrium in $F_i^+$.\hfill\(\blacksquare\)
    \end{proof}


    \begin{corollary} \label{Corollary 1}
    If the center is helpless, there isn't any protocol with truthfulness as a unique equilibrium.
    \end{corollary}
    
    Hence, we must assume for our desired protocol that the center is not helpless. The following Lemma demonstrates the strength of the competitive protocol. 
    
    \begin{lemma} \label{Lemma 3}
         The center is helpless if and only if $\;\forall s_1^r,\dots,s_N^r \in S^r: (s_1^r,\dots,s_N^r)$ is an equilibrium in the competitive protocol.
    \end{lemma}

    \begin{proof}[Direction 1]
    According to Lemma \ref{Lemma 2} the claim is true.\hfill\(\blacksquare\)
    \end{proof}

    \begin{proof}[Direction 2]
    $\forall s_1^r,\dots,s_N^r \in S^r: (s_1^r,\dots,s_N^r)$ is an equilibrium in the competitive protocol. Thus:
    \[\forall i \in \{1,\dots,N\}, \forall s_i \in S_i-S^r, \forall s_1^r,\dots,s_N^r \in S^r: \;
    U_i(s_i^r,s_{-i}^r;F_i^c) \geq U_i(s_i,s_{-i}^r;F_i^c) = U_i(s_i,s_{-i}^r;{F_i^c}^+).\]
    The equality is derived from the definition of ${F_i^c}^+$ and the fact that agent $i$ pays the most and therefore receives $b_0$ also in $F_i^c$.\\
    Now, let $A_i$ be an arbitrary protocol for agent $i$. Recall that $price_i^r(s_i^r)=0$ so even if agent $i$ pays the highest relative price, all the other agents pay 0 as well.
    Therefore, $\forall m_{-0}\in \{0,1\}^N: \; F_i^c(s_1^r,\dots,s_N^r,m)=\overrightarrow{0.5}$.\\
    Hence, according to Lemma \ref{Better than Nothing},
    $V_i(s_i^r,s_{-i}^r;A_i) \geq V_i(s_i^r,s_{-i}^r;F_i^c)$ so $U_i(s_i^r,s_{-i}^r;A_i) \geq U_i(s_i^r,s_{-i}^r;F_i^c)$.\\
    In addition, for each pair of protocols $P_i,Q_i$: \[\forall s_i \in S_i-S^r: U_i(s_i,s_{-i}^r;P_i^+)=U_i(s_i,s_{-i}^r;Q_i^+).\]
    It stems from the fact that all the information that the center has in such strategy profiles about the agents' private bits is her own bit (which she shares) and some information on $b_i$ (which the agent already knows).\\
    In particular,
    $U_i(s_i,s_{-i}^r;{F_i^c}^+)=U_i(s_i,s_{-i}^r;A_i^+)$.
    Therefore,
        \[U_i(s_i^r,s_{-i}^r;A_i) \geq U_i(s_i^r,s_{-i}^r;F_i^c) \geq U_i(s_i,s_{-i}^r;{F_i^c}^+) = U_i(s_i,s_{-i}^r;A_i^+)\] for every protocol $A_i$, meaning that the center is helpless.\hfill\(\blacksquare\)
    \end{proof}

    \begin{corollary} \label{Corollary 2}
    The profiles $\{(s_1^r,\dots,s_N^r)|s_1^r,\dots,s_N^r \in S^r\}$ are equilibria in the competitive protocol if and only if they are equilibria in every protocol.
    \end{corollary}

    The combination of Lemma \ref{Lemma 2} and Lemma \ref{Lemma 3} induces the above conclusion. The main use of these Lemmas is to show the equilibrium uniqueness of the protocol (under not helpless center). However, we notice that corollary \ref{Corollary 2} implies the protocol's optimality that we will show later in this paper.

    \begin{lemma} \label{Lemma 4}
        The competitive protocol has only two kinds of equilibria:  $(s^t,\dots,s^t)$ or $(s_1^r,\dots,s_N^r)$ $\forall s_1^r,\dots,s_N^r \in S^r$. 
    \end{lemma}

    \begin{proof}
    Let $g$ be the social function.
    Let $\vec{s}=(s_1,\dots,s_N)$ be an equilibrium in the competitive protocol assuming that $\vec{s} \neq (s^t,\dots,s^t),\; \vec{s} \notin \{(s_1^r,\dots,s_N^r)|s_1^r,\dots,s_N^r \in S^r\}$.
    
    Suppose that $\vec{s}$ contains at least two strategies ($s_i$ and $s_j$) that yield the same highest relative price (comparing to the others).
    In those cases, the competitive protocol submits $\overrightarrow{0.5}$ to all agents.
    Therefore, $\forall s_i^r \in S^r$:
        \[U_i(\vec{s};F_i^c)=V_i(\vec{s};F_i^c)-price_i(s_i)<V_i(\vec{s};F_i^c)=V_i(s_i^r,s_{-i};F_i^c)\\=U_i(s_i^r,s_{-i};F_i^c).\]
    So agent $i$ prefers to deviate.
    
    Now, suppose that $\vec{s}=(s_1,\dots,s_N)$ contains only one strategy $s_i$ that yields the highest relative price. Therefore agent $i$ receives
    $(b_0,m_1,\dots,m_N)$. Let $j$ be the agent who incurs the second highest relative price. She receives $\overrightarrow{0.5}$. Thus,
    if $s_j \in S^r$ then according to the intermediate value theorem (on $price_i^r$) there is $s'_i \in S_i$ such that $i$ prefers to deviate from $s_i$ to $s'_i$ and pay less but still receive $(b_0,m_1,\dots,m_N)$.
    Else, if $s_j \notin S^r$ then $j$ prefers to deviate to some $s'_j \in S^r$ (same considerations as before).
    
    Thus, we get that $\vec{s}$ is not an equilibrium.\\
    Note that deviating to $s'_j \in S^r$ is not necessarily $j$'s best response but it's sufficient for refuting the conjecture that $\vec{s}$ is an equilibrium.\hfill\(\blacksquare\)
    \end{proof}

    Lemma \ref{Lemma 4} is the core of the uniqueness proof. The competition between the agents and the main idea behind the construction of the competitive protocol is reflected in its proof.

    \begin{theorem} [Uniqueness] \label{Uniqueness Theorem}
    If the center is not helpless and truthfulness is an equilibrium in the competitive protocol, then it is the only equilibrium.
    \end{theorem}

    \begin{proof} [Uniqueness]
    Suppose that truthfulness is an equilibrium in the competitive protocol.
    According to Lemma \ref{Lemma 3} we get that $\forall s_1^r,\dots,s_N^r \in S^r: (s_1^r,\dots,s_N^r)$ is not an equilibrium in the protocol.
    Therefore, according to Lemma \ref{Lemma 4} we get that truthfulness is the only equilibrium.\hfill\(\blacksquare\)
    \end{proof}

    \subsection{Existence}
    In the previous sub-section we addressed the equilibrium uniqueness of the competitive protocol, assuming that truthfulness is indeed an equilibrium. Now, we approach the conditions in which truthfulness is indeed an equilibrium in this protocol. We start with a technical Lemma which implies that when the agents in the competitive protocol are truthful, then their best reply is either staying truthful or deviating to a zero information strategy. 
    

    \begin{lemma} \label{Lemma 6}
        \[\forall i\in \{1,\dots,N\}, \forall s_i^r\in S^r, \forall s_i\neq s^t\in S_i-S^r: U_i(s_i^r,(s^t,\dots,s^t);F_i^c) >  U_i(s_i,(s^t,\dots,s^t);F_i^c).\]
    \end{lemma}
    
    \begin{proof}
        Recall that $\forall s_i^r\in S^r, s_i\in S_i-S^r: price_i(s_i) > price_i(s_i^r) = 0$.\\
        In addition, if agent $i$ deviates from $s^t$ in the profile $(s^t,\dots,s^t)$ she receives $\overrightarrow{0.5}$ from the center regardless of the strategy she plays (there is at least one agent $j$ whose relative price is $price_j^r(s^t)=1$). Thus:
        \begin{multline*}
            \forall i\in \{1,\dots,N\}, \forall s_i^r\in S^r, \forall s_i\neq s^t\in S_i-S^r:\\ U_i(s_i^r,(s^t,\dots,s^t);F_i^c) = V_i(s_i^r,(s^t,\dots,s^t);F_i^c) > V_i(s_i^r,(s^t,\dots,s^t);F_i^c) - price_i(s_i) \\= V_i(s_i,(s^t,\dots,s^t);F_i^c) - price_i(s_i) =  U_i(s_i,(s^t,\dots,s^t);F_i^c).
        \end{multline*}\hfill\(\blacksquare\)
    \end{proof}

    We now approach the constraint on the existence of a truthful equilibrium in the competitive protocol. We will use the following notations:
    
    Let $F_i^{max}$ be a protocol such that:
    \[\forall s_1\in S_1,\dots,s_N\in S_N, \forall m_{-0}\in \{0,1\}^N: F_i^{max}(s_1,\dots,s_N,m)=(b_0,m_{-0}).\]
    Let $F_i^{min}$ be a protocol such that:
    \[\forall s_1\in S_1,\dots,s_N\in S_N, \forall m_{-0}\in \{0,1\}^N: F_i^{min}(s_1,\dots,s_N,m)=\overrightarrow{0.5}.\]

    \begin{definition} \label{Fanatic Definition}
    We call agent $i$ \emph{fanatic} if
        \[\forall s_1\in S_1,\dots,s_N\in S_N:\;
        price_i(s^t) > V_i(s^t,\dots,s^t;F_i^{max}) - V_i(s_1,\dots,s_N;F_i^{min}).\]
    \end{definition}

    A fanatic agent is an agent who won't be truthful at any rate. That is, her cost for truthfulness ($price_i(s^t)$) is greater than her marginal profit from full information about the social function $g$ ($F_i^{max}(s^t,\dots,s^t,m=b)$) above and beyond the value from knowing her own bit, $b_i$.

    \begin{lemma} \label{Lemma 7}
    If there is a fanatic agent then there is no protocol with a truthful equilibrium.
    \end{lemma}
    
    \begin{proof}
    Assume to the contrary that $i$ is a fanatic agent and $A_i$ is a protocol with a truthful equilibrium.
    Thus:
    \begin{multline*}
    \forall s_1\in S_1,\dots,s_N\in S_N, \forall s_i^r \in S^r:\; U_i(s^t,\dots,s^t;A_i) \geq U_i(s_i^r,(s^t,\dots,s^t);A_i) \\= V_i(s_i^r,(s^t,\dots,s^t);A_i) \geq V_i(s_1,\dots,s_N;F_i^{min}).
    \end{multline*}
    The last inequality is derived from Lemma \ref{Better than Nothing}.
    But,
    \begin{multline*}
    \forall s_1\in S_1,\dots,s_N\in S_N: \\
    U_i(s^t,\dots,s^t;A_i) = V_i(s^t,\dots,s^t;A_i) - price_i(s^t) \leq V_i(s^t,\dots,s^t;F_i^{max}) - price_i(s^t) \\ < V_i(s^t,\dots,s^t;F_i^{max}) - (V_i(s^t,\dots,s^t;F_i^{max}) - V_i(s_1,\dots,s_N;F_i^{min})) = V_i(s_1,\dots,s_N;F_i^{min}).
    \end{multline*}
    and this is a contradiction.\hfill\(\blacksquare\)
    \end{proof}
    
    \begin{lemma} \label{Lemma 8}
    If truthfulness is not an equilibrium in the competitive protocol then there is a fanatic agent.
    \end{lemma}
    
    \begin{proof}
    Assume that there is no fanatic agent.
      \[\forall i \in \{1,\dots,N\}, \exists s_1\in S_1,\dots,s_N\in S_N: \; V_i(s^t,\dots,s^t;F_i^{max}) - V_i(s_1,\dots,s_N;F_i^{min}) \geq price_i(s^t). \]
    Therefore,
    \begin{multline*}
        \forall i \in \{1,\dots,N\}, \exists s_1\in S_1,\dots,s_N\in S_N, \forall s_i^r\in S^r: \\ U_i(s^t,\dots,s^t;F_i^c) = V_i(s^t,\dots,s^t;F_i^c) - price_i(s^t) = V_i(s^t,\dots,s^t;F_i^{max}) - price_i(s^t) \\ \geq V_i(s_1,\dots,s_N;F_i^{min}) =
        V_i(s_i^r,(s^t,\dots,s^t);F_i^c) = U_i(s_i^r,(s^t,\dots,s^t);F_i^c).
    \end{multline*}
    In addition, according to Lemma \ref{Lemma 6}:
    \[\forall i\in \{1,\dots,N\}, \forall s_i^r\in S^r, \forall s_i\neq s^t\in S_i-S^r: U_i(s_i^r,(s^t,\dots,s^t);F_i^c) >  U_i(s_i,(s^t,\dots,s^t);F_i^c).\]
    yielding that \[\forall i \in \{1,\dots,N\}, \forall s_i\in S_i: U_i(s^t,\dots,s^t;F_i^c) \geq  U_i(s_i,(s^t,\dots,s^t);F_i^c).\]
    So the agents don't gain by deviating from $s^t$, meaning that truthfulness is an equilibrium in the competitive protocol.\hfill\(\blacksquare\) \end{proof}

    Notice that Lemma \ref{Lemma 8} proves the opposite direction of Lemma \ref{Lemma 7}. If there isn't any protocol with a truthful equilibrium, then in particular the competitive protocol doesn't have a truthful equilibrium, and according to Lemma \ref{Lemma 8} there is a fanatic agent. Therefore we get:

    \begin{corollary} \label{Corollary 4}
    There is a fanatic agent if and only if there is no protocol with a truthful equilibrium.
    \end{corollary}
    
    Similarly, Lemma \ref{Lemma 7} proves that if there is a fanatic agent then there isn't any protocol with a truthful equilibrium, and in particular there is no truthful equilibrium in the competitive protocol.  Combining with Lemma \ref{Lemma 8} we get:

    \begin{corollary} \label{Existence Corollary}
    Truthfulness is an equilibrium in the competitive protocol if and only if there is no fanatic agent.
    \end{corollary}


    

    \subsection{Optimality}
    
    Here we present a main result of the paper. We prove that if there is at least one protocol with a truthful equilibrium (not necessarily unique) then the competitive protocol has a truthful equilibrium as well.

    \begin{theorem} [Optimality] \label{Optimality Theorem}
    If there is a protocol with a truthful equilibrium, then truthfulness is an equilibrium in the competitive protocol.
    \end{theorem}

    \begin{proof} [Optimality]
    Assume by the way of contradiction that the profile $(s^t,\dots,s^t)$ is not an equilibrium in the competitive protocol. Thus, according to Lemma \ref{Lemma 6}, there is an agent $i$ such that
    $\forall s_i^r \in S^r:
    U_i(s_i^r,(s^t,\dots,s^t);F_i^c) >
    U_i(s^t,\dots,s^t;F_i^c)$.
    
    Let $A_i$ be a protocol with a truthful equilibrium. 
    We notice that 
    \begin{multline*}
        \forall k \in \{0,\dots,N\}: |E[p(b_k=1|F_i^{max}(s^t,\dots,s^t),m_i=b_i,b_i)]-0.5| = |E[p(b_k=1|b)]-0.5| \\= 0.5 \geq |E[p(b_k=1|A_i(s^t,\dots,s^t),m_i=b_i,b_i)]-0.5|.
    \end{multline*}
    Therefore, according to the comparison theorem $V_i(s^t,\dots,s^t;F_i^{max}) \geq V_i(s^t,\dots,s^t;A_i)$.
    Thus,
    $U_i(s^t,\dots,s^t;F_i^{max}) \geq U_i(s^t,\dots,s^t;A_i)$.\\
    In addition, according to the definition of $F_i^{min}$ and Lemma \ref{Better than Nothing} we get that $V_i(s_i^r,(s^t\dots,s^t);A_i) \geq  V_i(s_i^r,(s^t\dots,s^t);F_i^{min})$, and therefore $U_i(s^r,(s^t\dots,s^t);A_i) \geq  U_i(s^r,(s^t\dots,s^t);F_i^{min})$.
    
   Recall that $A_i$ has a truthful equilibrium.
    Thus,
    $\forall s_i^r \in S^r:
    U_i(s^t,\dots,s^t;A_i) \geq U_i(s_i^r,(s^t\dots,s^t);A_i).$
    Hence,
    $\forall s_i^r \in S^r:
    U_i(s^t,\dots,s^t;F_i^{max}) \geq U_i(s_i^r,(s^t\dots,s^t);F_i^{min})$.
    But recall that in the competitive protocol
    $U_i(s^t,\dots,s^t;F_i^c) = U_i(s^t,\dots,s^t;F_i^{max})$
    and
    $\forall s_i^r \in S^r:
    U_i(s_i^r,(s^t\dots,s^t);F_i^c) = U_i(s_i^r,(s^t\dots,s^t);F_i^{min}),$
    yielding that in the competitive protocol
    $\forall s_i^r \in S^r:
    U_i(s^t,\dots,s^t;F_i^c) \geq U_i(s_i^r,(s^t\dots,s^t);F_i^c)$ which is a contradiction.
    
    Thus, $(s^t,\dots,s^t)$ is an equilibrium in the competitive protocol as we wanted to show.\hfill\(\blacksquare\)
    \end{proof}

    Combining with the Uniqueness Theorem we get the following conclusion:
    \begin{corollary} \label{Optimality and Uniqueness}
    If the center is not helpless and there is a protocol with a truthful equilibrium (not necessarily unique) then the competitive protocol induces truthfulness as a unique equilibrium.
    \end{corollary}

    
    
    

    \section{Beyond Equilibrium - Fair Competitive Protocol}

     In this section we extend our domain of interest beyond equilibrium. We showed that the competitive protocol is optimal when considering the notion of truthful implementation in equilibrium. Here we present a slight improvement of the competitive protocol. We show a protocol possessing the properties of the competitive protocol which also guarantees information to the agents outside of equilibrium.
    Recall that the competitive protocol allows only truthful equilibrium (under some natural assumptions). Consequently, it may be inflexible: for instance, if all agents except one share their information truthfully, then they all receive nothing ($\overrightarrow{0.5}$).
    Based on this phenomenon we have an incentive to construct a moderate protocol that preserves uniqueness and optimality but also has some guarantee for truthful agents regardless of the others' strategies.

    \begin{definition} [The fair competitive protocol (FCP)]
    
    If an agent plays $s^t$ she receives all of the center's information $m=(b_0,m_{-0})$.
    The agent with the highest relative price among the non-truthful agents receives $m$ as long as she is not the only non-truthful agent. If at least two non-truthful agents incur the same highest relative price (among the non-truthful agents) then all the non-truthful agents receive $\overrightarrow{0.5}$. In all other cases the agent receives $\overrightarrow{0.5}$.
    Formally:
    \begin{multline*}
        \forall i\in \{1,\dots,N\}, \forall s_1\in S_1,\dots,s_N\in S_N, \forall m_{-0}\in \{0,1\}^N: \\F_i^{fc}(s_1,\dots,s_N,m)=
    \begin{cases}
    m & \text{if } \;\; s_i=s^t\\
    m & \text{if } \;\; price_i^r(s_i) > \underset{\{j|s_j\neq s^t,\; j\neq i\}}{max}\{price_j^r(s_j)\} \text{ and } s_{-i}\neq \overrightarrow{s^t}\\
    \overrightarrow{0.5} & \text{Otherwise}
    \end{cases}
    \end{multline*}
    \end{definition}
    
    
    It's easy to see that a truthful agent receives all the center's information regardless of the others. We show via equilibrium equivalence that the fair competitive protocol preserves the properties of the competitive protocol.
    
    \begin{theorem} [Equivalence] \label{Equivalence Theorem}
    The profile $\vec{s}=(s_1,\dots,s_N) \; s.t. \; s_1\in S_1,\dots s_N\in S_N$ is an equilibrium in the competitive protocol if and only if $\vec{s}$ is an equilibrium in the fair competitive protocol.
    \end{theorem}
    

    \begin{proof} [Direction 1]
    Let $\vec{s}=(s_1,\dots,s_N) \; s.t. \; s_1\in S_1,\dots s_N\in S_N$ and assume that $\vec{s}$ is an equilibrium in the competitive protocol.\\
    If $s_1,\dots,s_N\in S^r$ then according to Lemma \ref{Lemma 3} the center is helpless and therefore according to Lemma \ref{Lemma 2}, $\vec{s}$ is an equilibrium in FCP.
    
    If $s_1=\dots=s_N=s^t$ is an equilibrium in the competitive protocol then according to Corollary \ref{Corollary 4} there isn't a fanatic agent. 
    Let $i$ be an agent. Assume by the way of contradiction that there is $s_i\in S_i-s^t$ such that $U_i(s_i,(s^t,\dots,s^t);F_i^{fc}) > U_i(s^t,\dots,s^t;F_i^{fc})$.\\
    But,
    \begin{multline*}
        U_i(s_i,(s^t,\dots,s^t);F_i^{fc}) = V_i(s_i,(s^t,\dots,s^t);F_i^{fc}) - price_i(s_i) \\= V_i(s_i,(s^t,\dots,s^t);F_i^{min}) - price_i(s_i) \leq V_i(s_i,(s^t,\dots,s^t);F_i^{min}).
    \end{multline*}
    And,
    \[U_i(s^t,\dots,s^t;F_i^{fc}) = V_i(s^t,\dots,s^t;F_i^{fc}) - price_i(s^t) \\= V_i(s^t,\dots,s^t;F_i^{max}) - price_i(s^t).\]
    Hence,
    \begin{multline*}
        V_i(s_i,(s^t,\dots,s^t);F_i^{min}) > V_i(s^t,\dots,s^t;F_i^{max}) - price_i(s^t)\\ \Longrightarrow price_i(s^t) >V_i(s^t,\dots,s^t;F_i^{max}) - V_i(s_i,(s^t,\dots,s^t);F_i^{min}).
    \end{multline*}
    yielding that $i$ is a fanatic agent, which is a contradiction.
    Therefore $\vec{s}=(s^t,\dots,s^t)$ is an equilibrium in FCP.
    
    If $\vec{s}$ is not $(s^t,\dots,s^t)$ or in $\{(s_1^r,\dots,s_N^r)| s_1^r,\dots,s_N^r\in S^r\}$ then according to Lemma \ref{Lemma 4}, $\vec{s}$ is not an  equilibrium in the competitive protocol.\hfill\(\blacksquare\)
    \end{proof}

    \begin{proof} [Direction 2]
    Let $\vec{s}=(s_1,\dots,s_N) \; s.t. \; s_1\in S_1, \dots, s_N\in S_N$ and assume that $\vec{s}$ is an equilibrium in FCP. Let $i$ be an agent.\\
    If $s_1,\dots,s_N\in S^r$ then $U_i(s_1,\dots,s_N;F_i^c)=U_i(s_1,\dots,s_N;F_i^{fc})$.\\
    In addition, \[\forall s'_i\in S_i-S^r: U_i(s'_i,s_{-i};F_i^c)=U_i(s'_i,s_{-i};F_i^{fc}).\] But \[\forall s'_i\in S_i-S^r: U_i(s_1,\dots,s_N;F_i^{fc})\geq U_i(s'_i,s_{-i};F_i^{fc}).\] Hence, \[\forall s'_i\in S_i-S^r: U_i(s_1,\dots,s_N;F_i^c)\geq U_i(s'_i,s_{-i};F_i^c).\]
    If $s_1=\dots=s_N=s^t$ then according to the optimality theorem $\vec{s}$ is an equilibrium in the competitive protocol.
    
    Now, in order to complete the proof we will show that FCP doesn't have other equilibria.
    Assume by the way of contradiction that $\vec{s}$ is an equilibrium but neither truthful or zero-information.
    Therefore, there are at least two agents $l,k$ such that $s_l\neq s^t$ and $s_k\notin S^r$. 
    
    (1) Suppose that $s_k=s^t$ and suppose without loss of generality that $price_l^r(s_l) \geq \underset{\{j|s_j\neq s^t,\; j\neq l\}}{max}\{price_j^r(s_j)\}$.
    
    (1.1) If $price_l^r(s_l) = \underset{\{j|s_j\neq s^t,\; j\neq l\}}{max}\{price_j^r(s_j)\}$ then agent $l$ receives $\overrightarrow{0.5}$ and pays $price_l(s_l)$. Therefore, if $s_l\notin S^r$ then $price_l(s_l)>0$ so she prefers (not necessarily her best response) to deviate to some $s'_l\in S^r$ and pay 0 while still getting $\overrightarrow{0.5}$.\\
    Otherwise, $s_l\in S^r$ and therefore $price_l^r(s_l) = \underset{\{j|s_j\neq s^t,\; j\neq l\}}{max}\{price_j^r(s_j)\} = 0$, yielding that all the non-truthful agents play zero-information strategies (from $S^r$).
    In this case, agent $k$ must be truthful so $price_k^r(s_k=s^t)>0$ and agent $k$ receives $(b_0,m_{-0})$. Hence, according to the intermediate value theorem, agent $k$ prefers to deviate to some $s'_k \notin S^r$ such that $price_k^r(s_k)>price_k^r(s'_k)>0$ where she still receives $(b_0,m_{-0})$.
    
    (1.2) Otherwise, $price_l^r(s_l) > \underset{\{j|s_j\neq s^t,\; j\neq l\}}{max}\{price_j^r(s_j)\} \geq 0$.
    Again, according to the intermediate value theorem, agent $l$ prefers to deviate to some $s'_l\notin S^r$ such that $price_l^r(s_l)>price_l^r(s'_l)>\underset{\{j|s_j\neq s^t,\; j\neq l\}}{max}\{price_j^r(s_j)\}$ and pay less, while still receiving $(b_0,m_{-0})$.
    
    (2) Now, $s_k\neq s^t$ and recall that $s_k\notin S^r$. If $price_k^r(s_k) > \underset{\{j|s_j\neq s^t,\; j\neq k\}}{max}\{price_j^r(s_j)\}$ then according to the intermediate value theorem, agent $k$ prefers to deviate to some $s'_k\notin S^r$ such that $price_k^r(s_k)>price_k^r(s'_k)>\underset{\{j|s_j\neq s^t,\; j\neq k\}}{max}\{price_j^r(s_j)\}$ and still get $(b_0,m_{-0})$.
    Otherwise, agent $k$ pays $price_k(s_k)>0$ and gets $\overrightarrow{0.5}$ so she prefers to deviate to $s'_k \in S^r$ and pay 0 while still getting $\overrightarrow{0.5}$.
    
    In conclusion, FCP has two kinds of equilibria $(s^t,\dots,s^t)$ and $\{(s_1^r,\dots,s_N^r)| \\ s_1^r,\dots,s_N^r\in S^r\}$ and therefore each equilibrium in FCP is an equilibrium in the competitive protocol as well.\\
    From this equilibria equivalence between the competitive and the fair competitive protocols we get that all of our results about the competitive protocol apply also for the fair competitive protocol.\hfill\(\blacksquare\)
    \end{proof}

    \section{Future Work}
    
    In this work we modeled the setting of mediated multi-party computation with privacy-aware agents. Here we suggest several ways to expand our study.  

\begin{itemize}
    \item In our model we consider  a simple setting where each agent has a single bit, and those bits are i.i.d and uniformly distributed. One may generalize this setting to any distribution; this could be done quite directly by modifying the zero-information strategies. A more involved work should consider dropping the independence assumption. Additionally, the idea of our main result remains the same when allowing each agent to have more than one bit of information, so this may be generalized as well.

    \item Our main goal is that the social function will be computed correctly by the center and agents. Therefore, we search for a protocol which induces truthfulness as a single equilibrium. However, there may be cases where some agents posses information which doesn't affect the social function. In those scenarios, "forcing" an agent to reveal her private information doesn't help in computing the value of the social function. Hence, we could slightly change the protocol for these agents and let them keep their information confidential.

    \item  Recall that the condition for uniqueness is that the center should not be helpless. This condition is quite intuitive and well justified, as we show that if the center is helpless then no protocol can lead to a unique truthful equilibrium. An interesting followup is the classification of the set of social functions which induces the helplessness of the center. For instance, Example \ref{Helpless Example} shows a social function which yields a helpless center. Such classification would allow replacing the assumption required in order to obtain uniqueness by one which depends solely on the social function.

\end{itemize}

    \bibliographystyle{splncs04}
    \bibliography{bibdb.bib}

\begin{thebibliography}{10}
\providecommand{\url}[1]{\texttt{#1}}
\providecommand{\urlprefix}{URL }
\providecommand{\doi}[1]{https://doi.org/#1}

\bibitem{AbrahamDGH19}
Abraham, I., Dolev, D., Geffner, I., Halpern, J.Y.: Implementing mediators with
  asynchronous cheap talk. In: Robinson, P., Ellen, F. (eds.) Proceedings of
  the 2019 {ACM} Symposium on Principles of Distributed Computing, {PODC} 2019,
  Toronto, ON, Canada, July 29 - August 2, 2019. pp. 501--510. {ACM} (2019)

\bibitem{AbrahamDGH06}
Abraham, I., Dolev, D., Gonen, R., Halpern, J.Y.: Distributed computing meets
  game theory: robust mechanisms for rational secret sharing and multiparty
  computation. In: Ruppert, E., Malkhi, D. (eds.) Proceedings of the
  Twenty-Fifth Annual {ACM} Symposium on Principles of Distributed Computing,
  {PODC} 2006, Denver, CO, USA, July 23-26, 2006. pp. 53--62. {ACM} (2006)

\bibitem{aricha2013information}
Aricha, I., Smorodinsky, R.: Information elicitation and sequential mechanisms.
  International Journal of Game Theory  \textbf{42}(4),  931--946 (2013)

\bibitem{ashlagi2007k}
Ashlagi, I., Klinger, A., Tenneholtz, M.: K-ncc: Stability against group
  deviations in non-cooperative computation. In: International Workshop on Web
  and Internet Economics. pp. 564--569. Springer (2007)

\bibitem{chen2018optimal}
Chen, Y., Immorlica, N., Lucier, B., Syrgkanis, V., Ziani, J.: Optimal data
  acquisition for statistical estimation. In: Proceedings of the 2018 ACM
  Conference on Economics and Computation. pp. 27--44 (2018)

\bibitem{ChenK11}
Chen, Y., Kash, I.A.: Information elicitation for decision making. In:
  Sonenberg, L., Stone, P., Tumer, K., Yolum, P. (eds.) 10th International
  Conference on Autonomous Agents and Multiagent Systems {(AAMAS} 2011),
  Taipei, Taiwan, May 2-6, 2011, Volume 1-3. pp. 175--182. {IFAAMAS} (2011)

\bibitem{cummings2015accuracy}
Cummings, R., Ligett, K., Roth, A., Wu, Z.S., Ziani, J.: Accuracy for sale:
  Aggregating data with a variance constraint. In: Proceedings of the 2015
  conference on innovations in theoretical computer science. pp. 317--324
  (2015)

\bibitem{DworkMNS06}
Dwork, C., McSherry, F., Nissim, K., Smith, A.D.: Calibrating noise to
  sensitivity in private data analysis. In: Halevi, S., Rabin, T. (eds.) Theory
  of Cryptography, Third Theory of Cryptography Conference, {TCC} 2006, New
  York, NY, USA, March 4-7, 2006, Proceedings. Lecture Notes in Computer
  Science, vol.~3876, pp. 265--284. Springer (2006)

\bibitem{DworkRoth}
Dwork, C., Roth, A., et~al.: The algorithmic foundations of differential
  privacy. Found. Trends Theor. Comput. Sci.  \textbf{9}(3-4),  211--407 (2014)

\bibitem{GKmultisender}
Gentzkow, M., Kamenica, E.: Bayesian persuasion with multiple senders and rich
  signal spaces. Games and Economic Behavior  \textbf{104}(C),  411--429 (2017)

\bibitem{ghosh2015selling}
Ghosh, A., Roth, A.: Selling privacy at auction. Games and Economic Behavior
  \textbf{91},  334--346 (2015)

\bibitem{halpern2004rational}
Halpern, J., Teague, V.: Rational secret sharing and multiparty computation.
  In: Proceedings of the thirty-sixth annual ACM symposium on Theory of
  computing. pp. 623--632 (2004)

\bibitem{kairouz2021advances}
Kairouz, P., McMahan, H.B., Avent, B., Bellet, A., Bennis, M., Bhagoji, A.N.,
  Bonawitz, K., Charles, Z., Cormode, G., Cummings, R., et~al.: Advances and
  open problems in federated learning. Foundations and Trends{\textregistered}
  in Machine Learning  \textbf{14}(1--2),  1--210 (2021)

\bibitem{kamenica_gentzkow_2009}
Kamenica, E., Gentzkow, M.: Bayesian persuasion. American Economic Review
  \textbf{101},  2590--2615 (2011). \doi{10.3386/w15540}

\bibitem{mcgrew2003towards}
McGrew, R., Porter, R., Shoham, Y.: Towards a general theory of non-cooperative
  computation. In: Proceedings of the 9th conference on Theoretical aspects of
  rationality and knowledge. pp. 59--71 (2003)

\bibitem{procaccia2013approximate}
Procaccia, A.D., Tennenholtz, M.: Approximate mechanism design without money.
  ACM Transactions on Economics and Computation (TEAC)  \textbf{1}(4),  1--26
  (2013)

\bibitem{ronen2005prediction}
Ronen, A., Wahrmann, L.: Prediction games. In: International Workshop on
  Internet and Network Economics. pp. 129--140. Springer (2005)

\bibitem{shoham2005non}
Shoham, Y., Tennenholtz, M.: Non-cooperative computation: Boolean functions
  with correctness and exclusivity. Theoretical Computer Science
  \textbf{343}(1-2),  97--113 (2005)

\bibitem{smorodinsky2006overcoming}
Smorodinsky, R., Tennenholtz, M.: Overcoming free riding in multi-party
  computations—the anonymous case. Games and Economic Behavior
  \textbf{55}(2),  385--406 (2006)

\end{thebibliography}

    \appendix
    
    \section{A Preliminary Result}

    \begin{lemma}[Lemma \ref{Profit Lemma}] \label{Profit Lemma Appendix}
        \[V_i(\vec{s};F_i) = c_i\cdot \underset{b_{-i},m_{-i}}{E}[\underset{b_i,f_i,m_i}{E}[max\{p(g(b)=0|f_i,b_i,m_i),p(g(b)=1|f_i,b_i,m_i)\}]].\]
    \end{lemma}

    \begin{proof}
    For simplicity, we denote $I_i=\{b_i,m_i,f_i\}$.
    \begin{multline*}
        V_i(\vec{s};F_i) = \underset{b,m,f_i}{E}[v_i(b,m_i,f_i)] = \underset{b_{-i},m_{-i}}{E}[\underset{b_i,m_i,f_i}{E}[v_i(b,m_i,f_i)|b_i,m_i,f_i]] =\\ c_i\cdot \underset{b_{-i},m_{-i}}{E}[\underset{I_i}{E}[p(a_i(I_i)=g(b)|I_i)]] = c_i\cdot \underset{b_{-i},m_{-i}}{E}[\underset{I_i}{E}[p(a_i(I_i)=g(b)=1|I_i) + p(a_i(I_i)=g(b)=0|I_i)]].
    \end{multline*}
    \begin{multline*}
        p(a_i(I_i)=g(b)=1|I_i) = p(a_i(I_i)=1 \;\cap\; g(b)=1 |I_i) = p(a_i(I_i)=1 |I_i)\cdot p(g(b)=1 |I_i) \\= [p(p(g=1|I_i) > p(g=0|I_i)) + \frac{1}{2}p(p(g=1|I_i) = p(g=0|I_i))]\cdot p(g(b)=1 |I_i).
    \end{multline*}
    Notice that $a_i(I_i),\;g(b)$ are conditionally independent given $I_i$ because $a_i$ is solely determined by $I_i$, and $g$'s value doesn't depend on the action that the agent decides to make.
    Similarly,
    \[p(a_i(I_i)=g(b)=0|I_i) = [p(p(g=0|I_i) > p(g=1|I_i)) \\+ \frac{1}{2}p(p(g=1|I_i) = p(g=0|I_i))] \cdot p(g(b)=0 |I_i).\]
    Thus,
    \begin{multline*}
        V_i(\vec{s};F_i) = c_i\cdot \underset{b_{-i},m_{-i}}{E}[\underset{I_i}{E}[p(p(g=1|I_i) > p(g=0|I_i)) \cdot p(g(b)=1 |I_i) \\+ p(p(g=0|I_i) > p(g=1|I_i)) \cdot p(g(b)=0 |I_i) + \frac{1}{2}p(p(g=1|I_i) = p(g=0|I_i))\cdot (p(g=1|I_i) + p(g=0|I_i))]] \\= c_i\cdot \underset{b_{-i},m_{-i}}{E}[\underset{I_i}{E}[p(p(g=1|I_i) > p(g=0|I_i)) \cdot p(g(b)=1 |I_i) + p(p(g=0|I_i) > p(g=1|I_i)) \\\cdot p(g(b)=0 |I_i) + \frac{1}{2}p(p(g=1|I_i) = p(g=0|I_i))]]
    \end{multline*}
    But $p(g=0|I_i) + p(g=1|I_i) = 1$.
    Therefore,
    $p(g=0|I_i) > p(g=1|I_i)$ or $p(g=1|I_i) > p(g=0|I_i)$ or 
    $p(g=1|I_i) = p(g=0|I_i) = 0.5$.
    Hence,
    if $p(g=0|I_i) > p(g=1|I_i)$ then: 
    \begin{multline*}
        V_i(\vec{s};F_i) =
        c_i\cdot \underset{b_{-i},m_{-i}}{E}[\underset{I_i}{E}[0 \cdot p(g(b)=1 |I_i) + 1 \cdot p(g(b)=0 |I_i) + \frac{1}{2} \cdot 0]] = c_i\cdot \underset{b_{-i},m_{-i}}{E}[\underset{I_i}{E}[p(g(b)=0 |I_i)]].
    \end{multline*}
    Else, if $p(g=1|I_i) > p(g=0|I_i)$ then 
    \(V_i(p;F_i) = c_i\cdot \underset{b_{-i},m_{-i}}{E}\underset{I_i}{E}[p(g(b)=1 |I_i)]].\)\\
    Otherwise, if $p(g=1|I_i) = p(g=0|I_i)$ then 
    \[V_i(\vec{s};F_i) = c_i\cdot \underset{b_{-i},m_{-i}}{E}[\underset{I_i}{E}[\frac{1}{2}]] = c_i\cdot \underset{b_{-i},m_{-i}}{E}[\underset{I_i}{E}[p(g(b)=1 |I_i)]] = c_i\cdot \underset{b_{-i},m_{-i}}{E}[\underset{I_i}{E}[p(g(b)=0 |I_i)]].\]
    Therefore, in all possible cases we get that
    \[ V_i(\vec{s};F_i) = c_i\cdot \underset{b_{-i},m_{-i}}{E}[\underset{f_i,m_i,b_i}{E}[max\{p(g(b)=0|f_i,m_i,b_i),p(g(b)=1|f_i,m_i,b_i)\}]].\] \hfill\(\blacksquare\)
    \end{proof}
    
     \section{Informational Profitability}
    
    
    \begin{lemma}[Comparison Lemma] \label{Comparison Lemma Appendix}
        Let $i,j\in N$, where $i \neq j$.
        Assume that $Y^1_{-j} = Y^2_{-j}$ and $|Y_j^1-0.5| \geq |Y_j^2-0.5|$.
        Then,
        $V_i(\vec{s_1};F_i^1) \geq V_i(\vec{s_2};F_i^2)$.
    \end{lemma}
    
    
    \begin{proof}
    
    \[\forall val\in \{0,1\}, \forall f_i^1,m_i^1,b_i: p(g(b)=val|f_i^1,m_i^1,b_i) = p(g(b)=val|y^1,f_i^1,m_i^1,b_i).\]
    The equality holds due to the fact that $y^1$ is deterministic given $f_i^1,m_i^1,b_i$, and therefore the events $g(b)=val$ and $y^1$ are independent given $f_i^1,m_i^1,b_i$.
    
    In addition, $g$ is solely determined by $b$, meaning that $g(b)=val$ and $f_i^1,m_i^1$ are dependent through the relation of $b$ to $f_i^1,m_i^1$. But this connection is exactly $y^1$. $y^1$ represents the agent's inference on $b$ given $f_i^1,m_i^1,b_i$. 
    Thus, given $y^1 = (p(b_0=1|f_{i,0}^1,m_i,b_i),\dots,p(b_N=1|f_{i,N}^1,m_i,b_i))$, the events $g(b)=val$ and $\{f_i^1,m^1\}$ are independent. In addition, $b_i$ is exactly $y_i^1$. 
    Hence,
    \[\forall val\in \{0,1\}: p(g(b)=val|y^1,f_i^1,m_i^1,b_i) = p(g(b)=val|y^1,b_i) = p(g(b)=val|y^1).\]
    Therefore,
    \[p(g(b)=val|f_i^1,m_i^1,b_i) = p(g(b)=val|y^1),\]
    \begin{equation} \label{(1)}
        \underset{f_i^1,m_i^1,b_i}{E}[p(g(b)=val|f_i^1,m_i^1,b_i)] = \underset{f_i^1,m_i^1,b_i}{E}[p(g(b)=val|y^1)].
    \end{equation}
    Similarly,
    \begin{equation} \label{(2)}
        \underset{f_i^2,m_i^2,b_i}{E}[p(g(b)=val|f_i^2,m_i^2,b_i)] = \underset{f_i^2,m_i^2,b_i}{E}[p(g(b)=val|y^2)].
    \end{equation}

    According to the lemma's assumption, $|Y_j^1-0.5| \geq |Y_j^2-0.5|$.
    Therefore, there are four possible cases:\\
    
    1) $Y_j^1-0.5\geq 0,\; Y_j^2-0.5\geq 0$. $Y_j^1-0.5 \geq Y_j^2-0.5 \Longrightarrow Y_j^1\geq Y_j^2 \geq 0.5 \Longrightarrow \mathbf{Y_j^1\geq Y_j^2 \geq 1-Y_j^1}.$\\
    
    2) $Y_j^1-0.5\geq 0,\; Y_j^2-0.5\leq 0 \Longrightarrow Y_j^1\geq 0.5\geq Y_j^2.\; Y_j^1-0.5 \geq 0.5-Y_j^2 \Longrightarrow Y_j^2\geq 1-Y_j^1 \Longrightarrow \mathbf{Y_j^1 \geq Y_j^2\geq 1-Y_j^1}.$\\
    
    3) $Y_j^1-0.5\leq0,\; Y_j^2-0.5\geq0 \Longrightarrow Y_j^2\geq0.5\geq Y_j^1.\; 0.5-Y_j^1 \geq Y_j^2-0.5 \Longrightarrow Y_j^2\leq 1-Y_j^1 \Longrightarrow \mathbf{1-Y_j^1 \geq Y_j^2\geq Y_j^1}.$\\
    
    4) $Y_j^1-0.5\leq0,\; Y_j^2-0.5\leq0.\; 0.5-Y_j^1 \geq 0.5-Y_j^2 \Longrightarrow 0.5\geq Y_j^2\geq Y_j^1 \Longrightarrow \mathbf{1-Y_j^1\geq Y_j^2 \geq Y_j^1}.$\\
    
    We assume without loss of generality that $Y_j^1\geq Y_j^2$ and therefore $Y_j^1\geq \{ Y_j^2,\;1-Y_j^2 \} \geq 1-Y_j^1$ (otherwise we can change $Y^1,Y^2$ such that each coordinate $k$ represents the probability that the $k^{th}$ agent has the value $0$ instead of $1$).\\
    The segment $[1-Y_j^1,Y_j^1]$ is a convex set and therefore there is a $z\in [0,1]$ such that $Y_j^1\cdot z+(1-Y_j^1)\cdot (1-z) = Y_j^2$. 
    
    Recall that for any realization of $b,m,f$ the center has the ability to calculate $y_j^1$ and $y_j^2$. 
    Hence, instead of communicating to agent $i$ the reply $f_i^2$ in the regular way (which induces the probability of $y_j^2$), the center will calculate $y_j^1$ and give $f_i^1$ with the probability $z$ or give $\vec{1}-f_i^1$ with the probability $1-z$ (adding a noise of $z$ to the reply $f_i^1$).
    Notice that both ways are equivalent from the perspective of agent $i$ whose knowledge is $y_j^2$ at both.
    Thus, 
    \begin{equation} \label{(3)}
        \forall k \in \{0,1\}: E[p(g(b)=k|y^1,y^2)] = E[p(g(b)=k|y^1)].
    \end{equation}
    Assume without loss of generality that $E[p(g(b)=0|y^1)] \geq 0.5$.\\
    Therefore, from \ref{(1)},  $E[p(g(b)=0|b_i,m_i^1,f_i^1)] = E[p(g(b)=0|y^1)] \geq 0.5$.
    Hence, \begin{multline*}
        E[p(g(b) = a_i(b_i,m_i^1,f_i^1) = 0| b_i,m_i^1,f_i^1)] = E[p(g(b) = 0| b_i,m_i^1,f_i^1)] =\\ E[p(g(b) = 0| y^1)] = E[p(g(b) = 0| y^1,y^2)] = E[p(g(b)=0 \;\cap\; a_i(b_i,m_i^2,f_i^2) = 0|y^1, y^2)] +\\ E[p(g(b)=0 \;\cap\; a_i(b_i,m_i^2,f_i^2) = 1|y^1, y^2)].
    \end{multline*} The equalities are derived from the definition of $a_i$, \ref{(1)}, \ref{(3)}, and the law of total expectation respectively.
    In a similar way,
    \begin{multline*}
        E[p(g(b) = a_i(b_i,m_i^1,f_i^1) = 0| b_i,m_i^1,f_i^1)] = E[p(g(b) = 0| y^1)] \geq \\ E[p(g(b) = 1| y^1)] =  E[p(g(b) = 1| y^1,y^2)]= E[p(g(b)=1 \;\cap\; a_i(b_i,m_i^2,f_i^2) = 1|y^1, y^2)] +\\ E[p(g(b)=1 \;\cap\; a_i(b_i,m_i^2,f_i^2) = 0|y^1, y^2)]
    \end{multline*}
    and therefore,
    \begin{equation} \label{(4)}
        E[p(g(b) = 0|y^1)] \geq max\{E[p(g(b) = a_i(b_i,m_i^2,f_i^2) = 0|y^1, y^2)],\;E[p(g(b) = a_i(b_i,m_i^2,f_i^2) = 1|y^1, y^2)]\}
    \end{equation}
    yielding that
    \begin{multline*}
        V_i(\vec{s_1};F_i^1) = \underset{b_{-i},m_{-i}^1,m_{-i}^2}{E}[\underset{b_i,m_i^1,m_i^2,f_i^1,f_i^2}{E}[max\{p(g(b)=0|b_i,m_i^1,m_i^2,f_i^1,f_i^2),p(g(b)=1|b_i,m_i^1,m_i^2,f_i^1,f_i^2)\}]] \\= c_i\cdot \underset{b_{-i},m_{-i}^1,m_{-i}^2}{E}[\underset{b_i,m_i^1,m_i^2,f_i^1,f_i^2}{E}[max\{p(g(b)=0|y^1,y^2),p(g(b)=1|y^1,y^2)\}]] \\= c_i\cdot \underset{b_{-i},m_{-i}^1,m_{-i}^2}{E}[\underset{b_i,m_i^1,m_i^2,f_i^1,f_i^2}{E}[max\{p(g(b)=0|y^1),p(g(b)=1|y^1)\}]] \\= c_i\cdot \underset{b_{-i},m_{-i}^1,m_{-i}^2}{E}[\underset{b_i,m_i^1,m_i^2,f_i^1,f_i^2}{E}[p(g(b)=0|y^1)]]
        \\ \geq c_i\cdot
        \underset{b_{-i},m_{-i}^1,m_{-i}^2}{E}[\underset{b_i,m_i^1,m_i^2,f_i^1,f_i^2}{E}[max\{p(g(b) = a_i(b_i,m_i,f_i^2) = 0|y^1, y^2),p(g(b) = a_i(b_i,m_i,f_i^2) = 1|y^1, y^2)\}]] \\ = V_i(\vec{s_2};F_i^2)
    \end{multline*}
    The first and last equalities are derived from the law of total expectation. The second equality is derived from \ref{(1)} and \ref{(2)}. The third equality is derived from \ref{(3)}. The forth is derived from our assumption above, and the inequality is derived from \ref{(4)}. \hfill\(\blacksquare\)

    \end{proof}

    \begin{theorem}[Comparison Theorem] \label{Comparison Theorem Appendix}
    Assume that $\forall k\in \{0,\dots,N\}: |Y_k^1-0.5| \geq |Y_k^2-0.5|$. Then,
    $V_i(\vec{s_1};F_i^1) \geq V_i(\vec{s_2};F_i^2)$.
    \end{theorem}

    \begin{proof}
    Note that in order to ease the reading of the proof we abuse notations in the probabilistic expressions along the proof and condition on the protocols' distributions instead of the expectations of the center's replies (which are generated from these distributions). 
    Recall that $F_i^1=(F_{i,0}^1,\dots,F_{i,N}^1), F_i^2=(F_{i,0}^2,\dots,F_{i,N}^2)$.
    We construct a protocol $F_i$ such that $F_i(\vec{s_2},m^2)=(F_{i,0}^1(\vec{s_1},m_0^1),F_{i,-0}^2(\vec{s_2},m^2_{-0}))$. Hence, $E[p(b_0=1|F_{i,0}(\vec{s_2},m_0^2),m_i^2,b_i)] = Y_0^1$ and $\forall k\in \{1,\dots,N\}: E[p(b_k=1|F_{i,k}(\vec{s_2},m_k^2),m_i^2,b_i)] = Y_k^2$. Therefore, according to the comparison lemma $V_i(\vec{s_2};F_i) \geq V_i(\vec{s_2};F_i^2)$. Now, let $\vec{s'}=(s_1^1,s_2^2,\dots,s_N^2)$ be the strategy profile that yields the messages $m'=(b_0,m_1^1,m_2^2,\dots,m_N^2)$. In addition we construct $F'_i(\vec{s'},m')=(F_{i,0}(\vec{s_2},m_0^2), F_{i,1}^1(\vec{s_1},m_1^1)$,
    $F_{i,-\{0,1\}}(\vec{s_2},m_{-\{0,1\}}^2))$. Hence, $E[p(b_0=1|F'_{i,0}(\vec{s'},m'_0),m'_i,b_i)] = E[p(b_0=1|F_{i,0}(\vec{s_2},m_0^2),m_i^2,b_i)]$, $E[p(b_1=1|F'_{i,1}(\vec{s'},m'_1),m'_i,b_i)] = Y_0^1$ and $\forall k\in \{2,\dots,N\}: E[p(b_k=1|F'_{i,k}(\vec{s'},m'_k),m'_i,b_i)] = E[p(b_k=1|F_{i,k}(\vec{s_2},m_k^2),m_i^2,b_i)]$. Therefore, according to the comparison lemma $V_i(\vec{s'};F'_i) \geq V_i(\vec{s_2};F_i) \geq V_i(\vec{s_2};F_i^2)$. Now we construct $\vec{s''}=(s_1^1,s_2^1,s_3^2,\dots,s_N^2)$, $F''_i$ and show that $V_i(\vec{s''};F''_i) \geq V_i(\vec{s'};F'_i) \geq V_i(\vec{s_2};F_i^2)$. We continue this process for the rest of the $\{3,\dots,N\}$ bits (each time upgrading $Y_k^2$ to $Y_k^1$) until we get that $V_i(\vec{s_1};F_i^1) \geq V_i(\vec{s_2};F_i^2)$ as we wanted to show.\hfill\(\blacksquare\)
    \end{proof}

\section{Uniqueness}

    \begin{lemma} [Lemma \ref{Lemma 1}] \label{Lemma 1 Appendix}
        For every agent $i$ and her corresponding protocol $F_i$:
        \[\forall s_1\in S_1,\dots, s_N\in S_N: U_i(s_1,\dots,s_N;F_i^+) \geq U_i(s_1,\dots,s_N;F_i).\]
    \end{lemma}
    
    \begin{proof}
    Note that in order to ease the reading of the proof we abuse notations as we do in the proof of the Comparison Theorem \ref{Comparison Theorem Appendix}. Let $i$ be an agent, let $F_i$ be an arbitrary protocol, let $\vec{s}=(s_1,\dots,s_N)$ be a strategy profile and let $m$ be the agents' messages. If $F_i(\vec{s},m)=F_i^+(\vec{s},m)$ then obviously $U_i(\vec{s};F_i^+)=U_i(\vec{s};F_i)$.
    Otherwise, agent $i$ doesn't play a zero-information strategy and therefore she knows $b_0$ in $F_i^+$. Hence, $|E[p(b_0=1|F_i^+(\vec{s},m))]-0.5|=0.5\geq |E[p(b_0=1|F_i(\vec{s},m))]-0.5|$. In addition, notice that $F_{i,0}(\vec{s},m_0)$  and $\{b_1,\dots,b_N\}$ are independent (recall that the agents' bit are $i.i.d$) so the override of $F_{i,0}(\vec{s},m_0)$ in $F_i^+$ doesn't influence the agent's knowledge on $b_{-0}$. Combining with the fact that $F_{i,-0}=F_{i,-0}^+$ we get that her knowledge on the other bits is the same.
    Thus, according to the comparison lemma $V_i(\vec{s};F_i^+)\geq V_i(\vec{s};F_i)$, yielding that $U_i(\vec{s};F_i^+)\geq U_i(\vec{s};F_i)$.\hfill\(\blacksquare\)
    \end{proof}

    \begin{lemma} [Lemma \ref{Better than Nothing}] \label{Better than Nothing Appendix}
        Let $F_i^1,F_i^2$ be protocols for agent $i$. Let $\vec{s_1},\vec{s_2}$ and $m^1,m^2$ be the strategy profiles and the agents' messages in those protocols respectively. If $F_i^2(\vec{s_2},m^2)=\overrightarrow{0.5}$ then $V_i(\vec{s_1};F_i^1) \geq V_i(\vec{s_2};F_i^2)$.
    \end{lemma}
    
    \begin{proof}
    Note that in order to ease the reading of the proof we abuse notations as we do in the proof of the Comparison Theorem \ref{Comparison Theorem Appendix}.
    Notice that $\forall k\neq i \in \{0,\dots,N\}:
    |E[p(b_k=1|F_{i,k}^1(\vec{s_1},m_k^1),m_i^1,b_i)]-0.5| \geq 0 = |[E[p(b_k=1|F_{i,k}^2(\vec{s_2},m_k^2),m_i^1,b_i)]-0.5|$
    and obviously,
    $|E[p(b_i=1|F_{i,i}^1(\vec{s_1},m_i^1),m_i^2,b_i)]-0.5| = 0.5 = |E[p(b_i=1|F_{i,i}^2(\vec{s_2},m_i^2),m_i^2,b_i)]-0.5|$.
    Therefore, according to the comparison theorem $V_i(\vec{s_1};F_i^1) \geq  V_i(\vec{s_2};F_i^2)$.\hfill\(\blacksquare\)
    \end{proof}

\end{document}